\documentclass[letterpaper, 10 pt, conference]{ieeeconf}
\usepackage{amsmath,amsfonts,amssymb,mathtools}
\usepackage{graphicx}
\usepackage{BernsteinStyle}
\usepackage[hidelinks]{hyperref}
\usepackage{cite}

\IEEEoverridecommandlockouts

\newcommand{\LCSSLineStretch}{0.990}
\newcommand{\LCSSDisplaySkip}{3pt plus 1pt minus 1pt}
\newcommand{\LCSSShortDisplaySkip}{2pt plus 1pt minus 1pt}
\newcommand{\LCSSJot}{1.5pt}
\newcommand{\LCSSTextFloatSkip}{6pt plus 1pt minus 1pt}
\newcommand{\LCSSFloatSkip}{5pt plus 1pt minus 1pt}
\newcommand{\LCSSCaptionSkip}{2pt}
\newcommand{\LCSSSectionBefore}{1ex plus 0.5ex minus 0.2ex}
\newcommand{\LCSSSectionAfter}{0.5ex plus 0.2ex minus 0ex}
\newcommand{\LCSSSubsectionBefore}{0.9ex plus 0.4ex minus 0.2ex}
\newcommand{\LCSSSubsectionAfter}{0.45ex plus 0.2ex minus 0ex}

\makeatletter
\def\section{\@startsection{section}{1}{\z@}%
{\LCSSSectionBefore}{\LCSSSectionAfter}%
{\normalfont\normalsize\centering\scshape}}
\def\subsection{\@startsection{subsection}{2}{\z@}%
{\LCSSSubsectionBefore}{\LCSSSubsectionAfter}%
{\normalfont\normalsize\itshape}}
\makeatother

\newcounter{assumption}

\newenvironment{assumption}{
    \refstepcounter{assumption}
    \par\smallskip
    \noindent\textbf{Assumption~\theassumption.}
    \normalfont
}{
}

\title{\LARGE \bf
Taylor-Informed Indirect Adaptive Predictive Control \\ Using Jacobian-Frozen Affine Predictors
}

\author{Tam W. Nguyen$^{1}$
\thanks{$^{1}$Tam W. Nguyen is with the Department of Electrical Engineering, Kyoto University, Kyoto 615-8510, Japan {\tt\small nguyen.tamwilly.3e@kyoto-u.ac.jp}}%
}

\newtheorem{corollary}{Corollary}

\begin{document}

\maketitle
\thispagestyle{empty}
\pagestyle{empty}

\begin{abstract}
This paper develops a Taylor-informed indirect adaptive predictive control
framework for nonlinear sampled-data systems using Jacobian-frozen affine
predictors. A finite Taylor expansion approximates the sampled nonlinear
dynamics, and recursive least squares (RLS) identifies its polynomial
coefficients online. At each sampling instant, the Jacobian of the identified
map is evaluated at the current operating point and frozen over the prediction
horizon, yielding an affine predictor for model predictive control. In
contrast to generic nonlinear feature dictionaries, the implemented
polynomial dictionary is a forward-Euler/Taylor-structure-informed reduced
dictionary. Exact joint-odd symmetry eliminates even-total-degree monomials,
whereas additional forward-Euler-informed pruning constitutes a deliberate
model reduction. Numerical simulations on an unstable nonlinear benchmark
compare different Taylor degrees. The results show that higher-order models
improve tracking accuracy as the operating point moves farther from the
expansion point while maintaining comparable control effort. The complete
MATLAB implementation is publicly available to facilitate reproducibility.
\end{abstract}

\section{Introduction}

Indirect adaptive predictive control combines online identification with
model predictive control (MPC), adapting to streaming input--output data
while retaining the structure of linear predictive control
\cite{clarke1987generalized,clarke1987generalized2,rawlings2020model}. At each
sampling instant, the predictive model is updated recursively to compute a
finite-horizon control sequence. Recent developments have extended this
framework to aerospace and robotic applications
\cite{doi:10.2514/1.G008859,nguyen2026fast,doi:10.2514/1.G009500}.

Adaptive predictive control for nonlinear systems remains an active research
topic. Existing approaches include nonlinear model predictive control,
lifting based on observable dictionaries, and adaptive predictive methods
with nonlinear regressors
\cite{allgower2012nonlinear,kamaldar2026iterative,TAO2026113142,doi:10.2514/1.J066619,CASTROVIEJOFERNANDEZ2026115444}. More recently, nonlinear
extensions of the predictive cost adaptive control (PCAC) framework
\cite{9612636} have used kernel, polynomial, Fourier, and spline feature
dictionaries
\cite{nguyen2026adaptivebehavioralpredictivecontrol,
alhazmi2026nonlinearpredictivecostadaptive}.
These methods replace the linear regressor with a richer feature dictionary
whose coefficients are estimated through recursive least squares (RLS).
Polynomial, kernel, and other nonlinear observable representations can thus
increase expressiveness while preserving parameter linearity
\cite{engel2004kernel,liu2009extended,rosenfeld2024occupation}. Although these dictionaries often
improve approximation capability, they are often drawn from general basis
families whose connection to the underlying nonlinear plant is indirect.

This work takes a different approach. Instead of introducing a generic
nonlinear dictionary, we construct a reduced feature set from the Taylor
structure of the sampled-data dynamics and structural information from a
forward-Euler expansion. The resulting monomials have a clear analytical
interpretation. Exact symmetries of the sampled flow justify some
eliminations, while further omissions are imposed deliberately as
forward-Euler-informed model reduction. This additional reduction does not
imply that every omitted monomial vanishes in the sampled flow.

An exact nonlinear basis can be used directly when it is known and trusted.
In many applications, however, only partial structural knowledge is
available. A Taylor dictionary then provides a common local representation
without committing to a particular functional basis. Increasing its degree
accommodates richer local nonlinear behavior through coefficient adaptation
rather than regressor redesign.

After RLS identifies the polynomial coefficients online, the sampled
nonlinear map is converted into a locally affine predictive model. Its
Jacobian is evaluated at the current operating point, and the resulting
affine coefficients are frozen over the prediction horizon. This construction
avoids repeated nonlinear propagation while retaining online adaptation
through the identified Taylor coefficients. The predictor is then used in a
conventional MPC formulation.

Through numerical simulation, this paper investigates how the Taylor degree
affects closed-loop tracking. Structured Taylor dictionaries of different
orders are compared on a nonlinear benchmark to evaluate the trade-off
between model complexity and control performance. To facilitate
reproducibility, the complete MATLAB implementation used throughout this
study is publicly available at
\url{https://github.com/tamwng/taylor-informed-iapc}.

\section{System Description}

Consider the continuous-time nonlinear system
\begin{align}
\dot{x}=f(x,u),
\label{eq:ctsystem}
\end{align}
where $f:\mathbb{R}^n\times\mathbb{R}^m\to\mathbb{R}^n$,
$x\in\mathbb{R}^{n}$ is the state, and $u\in\mathbb{R}^{m}$ is the
control input.

Under zero-order hold with sampling period $T_s>0$,
\begin{align*}
u(t)=u_k,
\qquad
t\in[kT_s,(k+1)T_s),
\end{align*}
where $t\in\mathbb{R}_{\ge0}$ denotes continuous time and
$k\in\mathbb N_0$ is the sampling index. The exact sampled-data dynamics are
\begin{align}
x_{k+1}=F(x_k,u_k),
\label{eq:sampledmap}
\end{align}
where $F:\mathbb{R}^n\times\mathbb{R}^m\to\mathbb{R}^n$ denotes the
sampled-data flow map associated with \eqref{eq:ctsystem}.

Select an equilibrium pair $(x_{\rm e},u_{\rm e})$ satisfying
\begin{align}
f(x_{\rm e},u_{\rm e})=0.
\label{eq:equilibrium}
\end{align}
Then,
\begin{align*}
F(x_{\rm e},u_{\rm e})=x_{\rm e}.
\end{align*}
Define the shifted variables
\begin{align*}
\tilde{x} \coloneqq x-x_{\rm e}, \qquad
\tilde{u} \coloneqq u-u_{\rm e}.
\end{align*}
The shifted sampled-data dynamics are
\begin{align}
\tilde{x}_{k+1}
=
\bar{F}(\tilde{x}_k,\tilde{u}_k),
\label{eq:deviationmap}
\end{align}
where
\begin{align}
\bar{F}(\tilde{x},\tilde{u})
\coloneqq
F(x_{\rm e}+\tilde{x},
u_{\rm e}+\tilde{u})
-
x_{\rm e}.
\label{eq:shiftedmap}
\end{align}
By construction,
\begin{align}
\bar{F}(0,0)=0.
\label{eq:Fbar_equilibrium}
\end{align}

The following assumptions are used throughout this paper.

\begin{assumption}
The state $x_k$ is available for feedback at each time $k$.
\end{assumption}

\begin{assumption}
$\bar F$ is well-defined on a neighborhood
$\mathcal N\subset\mathbb R^{n+m}$ of the origin.
\end{assumption}

\begin{assumption}\label{ass:map_class}
$\bar F:\mathcal N\rightarrow\mathbb R^n$
is of class $\mathcal{C}^{D+1}$ on $\mathcal N$.
\end{assumption}

\section{Taylor Approximation of the Sampled Dynamics}

Under Assumption~\ref{ass:map_class}, the map
$\bar F:\mathcal N\rightarrow\mathbb R^n$ admits a finite Taylor
approximation about the origin.

Let
\begin{align}
p &\coloneqq n+m,
&
z
&\coloneqq
\begin{bmatrix}
\tilde x\\
\tilde u
\end{bmatrix}
\in\mathbb R^p.
\label{eq:z}
\end{align}
Since $\bar F(0)=0$, the order-$D$ Taylor approximation is
\begin{align}
\bar F(z)
=
Mz
+
\sum_{d=2}^{D}F_d(z)
+
r_D(z),
\label{eq:taylor_structure}
\end{align}
where
\begin{align}
M
\coloneqq
\left.
\frac{\partial \bar F}{\partial z}
\right|_{z=0}
\in\mathbb R^{n\times p},
\label{eq:M_linear}
\end{align}
and $r_D:\mathcal N\rightarrow\mathbb R^n$ is the Taylor remainder. For
each $d\in\{2,\ldots,D\}$, $F_d:\mathbb R^p\rightarrow\mathbb R^n$
collects the homogeneous polynomial terms of total degree $d$. Thus, $F_2$
contains all quadratic terms, $F_3$ contains all cubic terms, and so forth.
For $D=1$, the sum in \eqref{eq:taylor_structure} is empty.

The local Taylor approximation in \eqref{eq:taylor_structure} is a
polynomial in the components of $z$. Define the full Taylor index set
\begin{align}
\mathcal A_D
\coloneqq
\left\{
\alpha\in\mathbb N_0^p:
1\leq |\alpha|\leq D
\right\},
\qquad
|\alpha|
\coloneqq
\sum_{j=1}^{p}\alpha_j.
\label{eq:taylor_index_set}
\end{align}
For each $\alpha\in\mathcal A_D$, define the monomial
\begin{align}
z^\alpha
\coloneqq
\prod_{j=1}^{p} z_j^{\alpha_j},
\label{eq:monomial}
\end{align}
where $z_j$ is the $j$-th component of $z$.

Let
$\mathcal A_D=\{\alpha_1,\ldots,\alpha_q\}$.
The polynomial regressor is
\begin{align}
\phi_D(z)
\coloneqq
\begin{bmatrix}
z^{\alpha_1}
&
\cdots
&
z^{\alpha_q}
\end{bmatrix}^{\top}
\in\mathbb R^q.
\label{eq:taylor_regressor}
\end{align}
The regressor $\phi_D$ contains every pure and crossed monomial of total
degree from $1$ through $D$. Therefore,
\begin{align}
q
=
\binom{p+D}{D}-1.
\label{eq:num_monomials}
\end{align}

The order-$D$ Taylor model can therefore be written compactly as
\begin{align}
\bar F(z)
=
\Theta\phi_D(z)
+
r_D(z),
\label{eq:taylor_model}
\end{align}
where $\Theta\in\mathbb R^{n\times q}$ contains the Taylor coefficients.

\section{Recursive Taylor Coefficient Identification}

The Taylor model \eqref{eq:taylor_model} is linear in the coefficient
matrix $\Theta$. After neglecting the Taylor remainder, identification of
the sampled dynamics reduces to estimation of $\Theta$.

For each $k\geq0$, define
\begin{align}
z_k
\coloneqq
\begin{bmatrix}
\tilde x_k\\
\tilde u_k
\end{bmatrix},
\qquad
\phi_k
\coloneqq
\phi_D(z_k).
\label{eq:rls_phi_k}
\end{align}
The a priori one-step prediction of $\tilde x_{k+1}$ is
\begin{align}
\hat{\tilde x}_{k+1}
=
\hat\Theta_k\phi_k.
\label{eq:taylor_prediction}
\end{align}

Vectorize the coefficient matrix as
\begin{align}
\theta_k
\coloneqq
\operatorname{vec}(\hat\Theta_k)
\in\mathbb R^{nq},
\label{eq:theta_vec}
\end{align}
and define the regression matrix
\begin{align}
\varphi_k
\coloneqq
\phi_k^\top\otimes I_n
\in\mathbb R^{n\times nq}.
\label{eq:rls_regressor}
\end{align}
Then \eqref{eq:taylor_prediction} becomes
\begin{align}
\hat{\tilde x}_{k+1}
=
\varphi_k\theta_k.
\label{eq:vectorized_prediction}
\end{align}
The prediction error is
\begin{align}
e_k
\coloneqq
\tilde x_{k+1}
-
\varphi_k\theta_k.
\label{eq:rls_error}
\end{align}

Using standard RLS with forgetting factor $\lambda\in(0,1]$ and
$P_0\succ0$, define
\begin{align}
L_k
&=
\lambda^{-1}P_k,
\label{eq:rls_L}
\\
P_{k+1}
&=
L_k
-
L_k\varphi_k^\top
\left(
I_n+\varphi_kL_k\varphi_k^\top
\right)^{-1}
\varphi_kL_k,
\label{eq:rls_covariance}
\\
\theta_{k+1}
&=
\theta_k
+
P_{k+1}\varphi_k^\top e_k.
\label{eq:rls_parameter}
\end{align}

Processing the measured transition $(z_k,\tilde x_{k+1})$ yields the
posterior RLS estimate $(P_{k+1},\theta_{k+1})$. In the numerical study, these
one-step transitions come from the zero-order-hold RK4 simulation of the
continuous-time plant; forward Euler does not propagate the simulated plant.
Once $\tilde x_{k+1}$ is measured, $\theta_{k+1}$ is available for the
subsequent control synthesis.

Using this posterior estimate, reshape $\theta_{k+1}$ into
$\hat\Theta_{k+1}\in\mathbb R^{n\times q}$ and define the identified
sampled Taylor map
\begin{align}
\hat F_{k+1}(z)
\coloneqq
\hat\Theta_{k+1}\phi_D(z).
\label{eq:identified_taylor_map}
\end{align}
Thus, $\hat F_{k+1}$ follows from \eqref{eq:taylor_model} by neglecting the
Taylor remainder and replacing the unknown coefficient matrix $\Theta$ with
its RLS estimate $\hat\Theta_{k+1}$.

The map $\hat F_{k+1}$ is nonlinear in $z$ but linear in the identified
coefficients, which enables direct RLS identification. Other RLS variants,
including variable-rate forgetting or covariance reset, can be used without
changing the Jacobian-frozen construction below.

In the sequel, $\hat F_k$ denotes the most recent identified map
available at the current control instant.

\section{Jacobian-Frozen Affine Predictor}

Because the identified Taylor map $\hat F_k$ is nonlinear in the shifted
state and input, its exact finite-horizon propagation would lead to nonlinear
predictive control. We instead construct a local affine predictor by freezing
the Jacobian of $\hat F_k$ at the current operating point.

\subsection{Operating Point}

At step $k$, the shifted state $\tilde x_k$ is measured, $\tilde u_k$ is
currently applied, and the control computation produces $\tilde u_{k+1}$.
Accordingly, the operating point is
\begin{align}
\bar z_k
\coloneqq
z_k
=
\begin{bmatrix}
\tilde x_k\\
\tilde u_k
\end{bmatrix}
\in\mathbb R^{n+m}.
\label{eq:operating_point}
\end{align}
The initial input $u_0$ is prescribed.

The Taylor approximation remains centered at the equilibrium $z=0$; the
operating point $\bar z_k$ is used only to evaluate the Jacobian.

\subsection{Jacobian Evaluation}

Compute the Jacobian of the identified Taylor map at $\bar z_k$:
\begin{align}
J_k
\coloneqq
\left.
\frac{\partial \hat F_k}{\partial z}
\right|_{z=\bar z_k}
\in\mathbb R^{n\times(n+m)}.
\label{eq:jacobian_full}
\end{align}
Partition
\begin{align}
J_k
=
\begin{bmatrix}
A_k & B_k
\end{bmatrix},
\label{eq:jacobian_partition}
\end{align}
where
$A_k\in\mathbb R^{n\times n}$ and
$B_k\in\mathbb R^{n\times m}$ are the Jacobians of
$\hat F_k$ with respect to the shifted state and input,
respectively.

The first-order local approximation of $\hat F_k$ at $\bar z_k$ is
\begin{align}
\hat F_k(z)
&\approx
\hat F_k(\bar z_k)
+
J_k(z-\bar z_k).
\label{eq:local_linearization}
\end{align}
Substituting
$z=[\tilde x^\top\ \tilde u^\top]^\top$
and partitioning $J_k$ according to
\eqref{eq:jacobian_partition} yields
\begin{align}
\hat F_k(z)
&\approx
c_k+A_k\tilde x+B_k\tilde u,
\label{eq:affine_predictor}
\end{align}
where
\begin{align}
c_k
\coloneqq
\hat F_k(\bar z_k)
-
A_k\tilde x_k
-
B_k\tilde u_k.
\label{eq:predictor_offset}
\end{align}

The resulting Jacobian-frozen predictor is
\begin{align}
\tilde x_{i+1|k}
=
c_k
+
A_k\tilde x_{i|k}
+
B_k\tilde u_{i|k},
\quad
i=0,\ldots,N-1.
\label{eq:jf_predictor}
\end{align}
The matrices $A_k$, $B_k$, and $c_k$ remain fixed over the prediction
horizon.

\subsection{Reduction to Linear PCAC}

Linear PCAC is recovered as the first-order special case of the proposed
scheme.

\begin{corollary}
If $D=1$, then the Jacobian-frozen predictor
\eqref{eq:jf_predictor} reduces to the linear PCAC propagation model.
\end{corollary}

\begin{proof}
For $D=1$, the Taylor regressor satisfies $\phi_1(z)=z$. Hence,
\begin{align}
\hat F_k(z)
=
\hat\Theta_k z,
\end{align}
where $\hat\Theta_k\in\mathbb R^{n\times(n+m)}$. Therefore,
\begin{align}
J_k
=
\hat\Theta_k
=
\begin{bmatrix}
A_k & B_k
\end{bmatrix}
\end{align}
is independent of the operating point. Moreover,
\begin{align}
c_k
&=
\hat F_k(\bar z_k)-J_k\bar z_k
\\
&=
\hat\Theta_k\bar z_k-\hat\Theta_k\bar z_k
=
0.
\end{align}
Substituting $c_k=0$ into \eqref{eq:jf_predictor} yields
\begin{align}
\tilde x_{i+1|k}
=
A_k\tilde x_{i|k}
+
B_k\tilde u_{i|k},
\end{align}
which is the linear PCAC propagation model.
\end{proof}

\subsection{Elimination of Factor-Freezing Ambiguity}

Jacobian freezing avoids arbitrary choices when approximating mixed nonlinear
terms. For example, consider
\begin{align}
F(y,u)
=
\alpha y+\beta u+a u^2+b uy+c y^2.
\label{eq:quadratic_example}
\end{align}
At $(\bar y,\bar u)$, the Jacobian-frozen predictor is
\begin{align}
F(y,u)
\approx
C_{\rm loc}+A_{\rm loc} y+B_{\rm loc} u,
\end{align}
where
\begin{align}
A_{\rm loc}
&=
\alpha+b\bar u+2c\bar y,
\\
B_{\rm loc}
&=
\beta+2a\bar u+b\bar y,
\\
C_{\rm loc}
&=
F(\bar y,\bar u)-A_{\rm loc}\bar y-B_{\rm loc}\bar u.
\end{align}
Thus, the cross term satisfies
\begin{align}
uy
\approx
\bar u\,y+\bar y\,u-\bar u\,\bar y.
\end{align}
This approximation matches both the value and the Jacobian at the operating
point without requiring a choice of whether to freeze $u$ or $y$.

\section{Predictive Control Synthesis}

Following the PCAC formulation \cite{9612636}, the Jacobian-frozen predictor
\eqref{eq:jf_predictor} is applied over a prediction horizon $N$.

At step $k$, initialize
\begin{align}
\tilde x_{0|k}
&=
\tilde x_k,
&
\tilde u_{0|k}
&=
\tilde u_k.
\label{eq:prediction_initialization}
\end{align}
The predicted states satisfy
\begin{align}
\tilde x_{i+1|k}
=
c_k
+
A_k\tilde x_{i|k}
+
B_k\tilde u_{i|k},
\quad
i=0,\ldots,N-1.
\label{eq:prediction_dynamics}
\end{align}
Here, $\tilde u_{0|k}$ is known, whereas
$\tilde u_{1|k},\ldots,\tilde u_{N-1|k}$ are decision variables. This
indexing follows the PCAC convention: $\tilde u_{1|k}$ is applied at step
$k+1$.

Let
\begin{align}
\Delta\tilde u_{i|k}
\coloneqq
\tilde u_{i|k}-\tilde u_{i-1|k},
\qquad
i=1,\ldots,N-1,
\label{eq:input_increment}
\end{align}
and let $\tilde r_{i|k}$ denote the reference in shifted coordinates.

Let the state and input constraint sets be the convex polyhedra
\begin{align}
\mathcal X
&=
\left\{
\tilde x\in\mathbb R^n:
H_x\tilde x\leq h_x
\right\},
\\
\mathcal U
&=
\left\{
\tilde u\in\mathbb R^m:
H_u\tilde u\leq h_u
\right\}.
\end{align}
To soften the state constraints, introduce
$\varepsilon_{i|k}\geq0$ and solve
\begin{align}
\underset{
\substack{
\tilde x_{1|k},\ldots,\tilde x_{N|k},\\
\tilde u_{1|k},\ldots,\tilde u_{N-1|k},\\
\varepsilon_{1|k},\ldots,\varepsilon_{N|k}
}}
{\operatorname{minimize}}
\quad
&
\frac{1}{2}
\sum_{i=1}^{N}
\left\|
\tilde x_{i|k}-\tilde r_{i|k}
\right\|_{Q_i}^{2}
\notag\\
&+
\frac{1}{2}
\sum_{i=1}^{N-1}
\left\|
\Delta\tilde u_{i|k}
\right\|_{R}^{2}
+
\frac{1}{2}
\sum_{i=1}^{N}
\left\|
\varepsilon_{i|k}
\right\|_{S}^{2}
\label{eq:predictive_qp}
\\
\operatorname{subject\ to}
\quad
&
\eqref{eq:prediction_dynamics},
\notag\\
&
H_x\tilde x_{i|k}
\leq
h_x+\varepsilon_{i|k},
\qquad
i=1,\ldots,N,
\notag\\
&
H_u\tilde u_{i|k}
\leq
h_u,
\qquad
i=1,\ldots,N-1,
\notag\\
&
\varepsilon_{i|k}\geq0,
\qquad
i=1,\ldots,N,
\notag
\end{align}
where
$Q_i\succeq0$,
$R\succ0$,
$S\succ0$, and
$\|v\|_Q^2\coloneqq v^\top Qv$.

Because $A_k$, $B_k$, and $c_k$ are fixed over the horizon,
\eqref{eq:predictive_qp} is a convex quadratic program. The slack variables
prevent infeasibility caused solely by the state constraints but permit
violation of the nominal state constraint. No recursive-feasibility or
closed-loop stability guarantee is asserted because the identified model and
its Jacobian-frozen predictor vary online.

Let $\tilde u_{1|k}^{\star}$ denote the first optimized input. The computed
control is
\begin{align}
u_{k+1}
=
u_{\rm e}
+
\tilde u_{1|k}^{\star}.
\label{eq:applied_control}
\end{align}
Thus, $u_k$ generates the measured transition from $k$ to $k+1$, while the
newly computed $u_{k+1}$ is held over the subsequent sampling interval.

\section{Case Study: Unstable Trigonometric Plant}
\label{sec:case_study}

\subsection{Plant and Sampled Dynamics}
\label{subsec:case_plant}

Consider the continuous-time scalar system
\begin{align}
    \dot{x}
    =
    f(x,u)
    :=
    a\sin x\cos x
    +
    b\cos^2 x\sin u,
    \label{eq:case_plant}
\end{align}
where $a>0$ and $b>0$. We select $(x_{\rm e},u_{\rm e})=(0,0)$ as the
equilibrium. Hence, $\tilde{x}=x$ and $\tilde{u}=u$, and tildes are omitted
throughout this section. Near the origin,
\begin{align}
    \dot{x}
    =
    ax+bu+\mathcal{O}\left(\|(x,u)\|^3\right),
    \label{eq:case_local_dynamics}
\end{align}
so the equilibrium is open-loop unstable and its linearization is
controllable.

Under zero-order hold with sampling period $T_s=0.05~{\rm s}$, let $F$
denote the exact sampled-data flow map. In the numerical study, an RK4 step
propagates the continuous-time plant over each sampling interval under a
constant input. The RLS estimator identifies the resulting one-step
transitions; forward Euler does not propagate the simulated plant.

For analytical reference only, applying a forward-Euler step to
\eqref{eq:case_plant} gives
\begin{align}
    x_{k+1}
    &=
    F_{\rm FE}(x_k,u_k)
    \nonumber\\
    &:=
    x_k
    +
    T_s\left(
        a\sin x_k\cos x_k
        +
        b\cos^2 x_k\sin u_k
    \right).
    \label{eq:case_forward_euler_reference}
\end{align}
Its linear expansion near the origin is
\begin{align}
    x_{k+1}
    =
    (1+aT_s)x_k+bT_su_k
    +
    \mathcal{O}\left(\|(x_k,u_k)\|^3\right),
    \label{eq:case_discrete_linearization}
\end{align}
whose open-loop eigenvalue is $1+aT_s>1$.

Since
\begin{align}
    f(-x,-u)=-f(x,u),
\end{align}
uniqueness of solutions implies the exact joint-odd symmetry of the sampled
flow,
\begin{align*}
    F(-x,-u)
    =
    -F(x,u).
\end{align*}
The identity also holds exactly for the implemented RK4 map because every
stage preserves the joint sign reversal. Consequently, neither the sampled
flow nor the RK4 map has even-total-degree terms in its Taylor expansion about
the origin.

\subsection{Taylor Structure and Coefficient Pruning}
\label{subsec:case_taylor_structure}

Let
\begin{align}
    z=
    \begin{bmatrix}
        x & u
    \end{bmatrix}^\top .
\end{align}
As an analytical reference, expanding the forward-Euler expression
\eqref{eq:case_forward_euler_reference} about the origin through degree five
yields
\begin{align}
    F_{\rm FE}(x,u)
    ={}&
    (1+aT_s)x+bT_su
    \nonumber\\
    &-\frac{2aT_s}{3}x^3
    -bT_sx^2u
    -\frac{bT_s}{6}u^3
    \nonumber\\
    &+\frac{2aT_s}{15}x^5
    +\frac{bT_s}{3}x^4u
    +\frac{bT_s}{6}x^2u^3
    +\frac{bT_s}{120}u^5
    \nonumber\\
    &+\mathcal{O}\left(\|z\|^7\right).
    \label{eq:case_fifth_order_expansion}
\end{align}

Joint-odd symmetry removes every even-total-degree term exactly. The
implemented dictionary then applies a deliberate additional model reduction
informed by the forward-Euler structure. For a degree-$D$ truncation, it
retains only monomials of the form
\begin{align}
    x^{2i+1},
    \qquad
    x^{2i}u^{2j+1},
    \qquad
    i,j\in\mathbb{N}_0,
\end{align}
with total degree at most $D$. This additional pruning does not imply that
every omitted odd-total-degree monomial vanishes in the RK4 sampled flow.
The resulting feature set is therefore called a
forward-Euler/Taylor-structure-informed reduced dictionary.

The resulting regressors are
\begin{align*}
    \varphi^{\rm FE}_1
    &=
    \begin{bmatrix}
        x &
        u
    \end{bmatrix}^\top,
    \\
    \varphi^{\rm FE}_3
    &=
    \begin{bmatrix}
        x &
        u &
        x^3 &
        x^2u &
        u^3
    \end{bmatrix}^\top,
    \\
    \varphi^{\rm FE}_5
    &=
    \begin{bmatrix}
        x &
        u &
        x^3 &
        x^2u &
        u^3 &
        x^5 &
        x^4u &
        x^2u^3 &
        u^5
    \end{bmatrix}^\top,
    \\
    \varphi^{\rm FE}_7
    &=
    \begin{bmatrix}
        \left(\varphi^{\rm FE}_5\right)^\top &
        x^7 &
        x^6u &
        x^4u^3 &
        x^2u^5 &
        u^7
    \end{bmatrix}^\top .
\end{align*}

Table~\ref{tab:case_coefficient_counts} compares the coefficient counts of
the full Taylor dictionary and its implemented reduction.
\begin{table}[t]
    \centering
    \caption{Number of estimated coefficients in the full and implemented
    reduced dictionaries.}
    \label{tab:case_coefficient_counts}
    \begin{tabular}{c|c|c}
        \hline
        Degree $D$
        &
        Full dictionary
        &
        Reduced dictionary
        \\
        \hline
        1 & 2  & 2  \\
        3 & 9  & 5  \\
        5 & 20 & 9  \\
        7 & 35 & 14 \\
        \hline
    \end{tabular}
\end{table}

For analytical reference, \eqref{eq:case_fifth_order_expansion} gives the
coefficient vector associated with $\varphi^{\rm FE}_5$ as
\begin{align}
    \theta_5^\star
    =
    \left[
    \begin{smallmatrix}
        1+aT_s\\
        bT_s\\
        -2aT_s/3\\
        -bT_s\\
        -bT_s/6\\
        2aT_s/15\\
        bT_s/3\\
        bT_s/6\\
        bT_s/120
    \end{smallmatrix}
    \right].
    \label{eq:case_true_theta5}
\end{align}

The forward-Euler-informed model reduction therefore decreases the number
of estimated coefficients from $20$ to $9$ for $D=5$ and from $35$ to $14$
for $D=7$. The vector $\theta_5^\star$ is only an analytical forward-Euler
reference: the controller neither uses it nor compares it numerically with
the identified RK4-transition coefficients. Instead, RLS estimates the
retained coefficients online.

No constant monomial is estimated. The affine term arises subsequently from
Jacobian freezing.

\subsection{Identified Map and Jacobian-Frozen Predictor}
\label{subsec:case_compilation}

For each Taylor degree $D\in\{1,3,5,7\}$, let $\varphi_D^{\rm FE}$ denote
the implemented reduced regressor defined by the preceding monomial pattern.
The identified sampled map is
\begin{align}
    \widehat{F}_{D,k}(x,u)
    =
    \widehat{\theta}_{D,k}^\top
    \varphi_D^{\rm FE}(x,u),
    \label{eq:case_identified_map}
\end{align}
where $\widehat{\theta}_{D,k}$ is updated online using RLS.

At each step $k$, the Jacobian-frozen coefficients are evaluated at the
current operating point $(x_k,u_k)$:
\begin{align}
    A_k
    &=
    \left.
    \frac{\partial \widehat{F}_{D,k}}{\partial x}
    \right|_{(x_k,u_k)},
    \label{eq:case_Ak_general}
    \\
    B_k
    &=
    \left.
    \frac{\partial \widehat{F}_{D,k}}{\partial u}
    \right|_{(x_k,u_k)},
    \label{eq:case_Bk_general}
    \\
    c_k
    &=
    \widehat{F}_{D,k}(x_k,u_k)
    -A_kx_k-B_ku_k.
    \label{eq:case_ck_general}
\end{align}

For the fifth-order reduced regressor $\varphi_5^{\rm FE}$, let
$\widehat{\theta}_{ij,k}$ denote the estimated coefficient multiplying
$x^iu^j$. The resulting Jacobian coefficients are
\begin{align}
    A_k
    ={}&
    \widehat{\theta}_{10,k}
    +3\widehat{\theta}_{30,k}x_k^2
    +2\widehat{\theta}_{21,k}x_ku_k
    \nonumber\\
    &+
    5\widehat{\theta}_{50,k}x_k^4
    +4\widehat{\theta}_{41,k}x_k^3u_k
    +2\widehat{\theta}_{23,k}x_ku_k^3,
    \label{eq:case_Ak_D5}
    \\
    B_k
    ={}&
    \widehat{\theta}_{01,k}
    +\widehat{\theta}_{21,k}x_k^2
    +3\widehat{\theta}_{03,k}u_k^2
    \nonumber\\
    &+
    \widehat{\theta}_{41,k}x_k^4
    +3\widehat{\theta}_{23,k}x_k^2u_k^2
    +5\widehat{\theta}_{05,k}u_k^4.
    \label{eq:case_Bk_D5}
\end{align}

The resulting Jacobian-frozen predictor is
\begin{align}
    x_{i+1|k}
    =
    c_k+A_kx_{i|k}+B_ku_{i|k},
    \quad
    i=0,\ldots,N-1.
    \label{eq:case_frozen_predictor}
\end{align}
For analytical reference only, the Jacobian of the forward-Euler expression
\eqref{eq:case_forward_euler_reference} is
\begin{align}
    A_{\rm FE}(x,u)
    &=
    1
    +
    T_s
    \left[
        a\cos(2x)
        -
        b\sin(2x)\sin u
    \right],
    \label{eq:case_true_A}
    \\
    B_{\rm FE}(x,u)
    &=
    T_sb\cos^2x\cos u.
    \label{eq:case_true_B}
\end{align}

\section{Numerical Evaluation}
\label{sec:numerical_evaluation}

This section evaluates how the Taylor degree affects closed-loop tracking.
The continuous-time plant \eqref{eq:case_plant} is simulated under zero-order
hold using fourth-order Runge--Kutta (RK4) with $T_s=0.05~{\rm s}$ over each
sampling interval. RLS identifies the one-step transitions generated by this
RK4 simulation. The controller combines the identified
forward-Euler/Taylor-structure-informed reduced map with the Jacobian-frozen
predictor developed above. All cases $D\in\{1,3,5,7\}$ use the same
initialization data, constraints, prediction horizon, and control weights.

The complete MATLAB implementation used throughout this study is publicly available at \url{https://github.com/tamwng/taylor-informed-iapc}.

\subsection{Simulation Setup}
\label{subsec:numerical_configuration}

The numerical configuration is summarized in
Table~\ref{tab:numerical_configuration}.
\begin{table}[t]
    \centering
    \caption{Numerical configuration.}
    \label{tab:numerical_configuration}
    \begin{tabular}{l|c}
        \hline
        Quantity & Value \\
        \hline
        Plant parameters
        & $a=0.5,\ b=1.5$ \\
        Sampling period
        & $T_s=0.05~{\rm s}$ \\
        Plant integration
        & RK4 with ZOH input \\
        Initial state
        & $x_0=0$ \\
        Constraints
        & soft $|x_k|\leq0.75$, hard $|u_k|\leq1$ \\
        Taylor degrees
        & $D\in\{1,3,5,7\}$ \\
        Prediction horizon
        & $N=8$ \\
        Forgetting factor
        & $\lambda=1$ \\
        Initial RLS estimate
        & $\widehat{\theta}_{D,0}=0_{q_D^{\rm FE}}$ \\
        Initial covariance
        & $P_0=10^5I_{q_D^{\rm FE}}$ \\
        State weights
        & $Q_i=1,\ i<N,\quad Q_N=10$ \\
        Input-increment weight
        & $R=5\times10^{-2}$ \\
        State-slack weight
        & $S=10^5$ \\
        Initialization length
        & $N_{\rm id}=60$ samples \\
        Initialization input
        & $u_k\in\{0,\pm0.15,\pm0.30,\pm0.45,\pm0.60\}$ \\
        Reference type
        & Amplitude-swept sinusoid \\
        Segment length
        & $K_r=300$ samples \\
        Sinusoidal frequency
        & $0.10~{\rm Hz}$ \\
        Amplitudes
        & $\{0.25,\ 0.55,\ 0.80,\ 0.90,\ 0\}$ \\
        \hline
    \end{tabular}
\end{table}
During the first $N_{\rm id}$ samples, a fixed multilevel input sequence is
applied while the RLS estimator is updated. These initialization levels
excite the odd monomials retained in the reduced dictionary. The predictive
controller is activated at $k=N_{\rm id}$. The state bound
$|x_k|\leq0.75$ is softened by the slack variables in
\eqref{eq:predictive_qp} and may therefore be violated; the input bound is
hard.

After the initialization phase, the reference is an amplitude-swept
sinusoid. For the $j$th segment, the reference is defined by
\begin{align*}
    r_k
    =
    A_j
    \sin\!\left(
        2\pi f
        (k-N_{\rm id})T_s
    \right),
    \label{eq:numerical_reference}
\end{align*}
where $f=0.10~{\rm Hz}$ and
\begin{align*}
A_j\in\{0.25,\ 0.55,\ 0.80,\ 0.90,\ 0\}.
\end{align*}
Each amplitude is maintained for $K_r=300$ samples before transitioning to
the next level. The increasing amplitudes progressively excite higher-order
nonlinearities. Because the state constraint is soft, the closed-loop state
can exceed its nominal bound.

The cases are compared through state and input trajectories, tracking errors,
and quantitative performance metrics computed over each amplitude segment.

\subsection{Closed-Loop Results}
\label{subsec:closed_loop_results}

Figures~\ref{fig:tracking}--\ref{fig:amplitude_error} summarize the
closed-loop responses produced by the proposed controller.
\begin{figure}
    \centering
    \includegraphics[width=\linewidth, trim={10 0 30 0},clip]{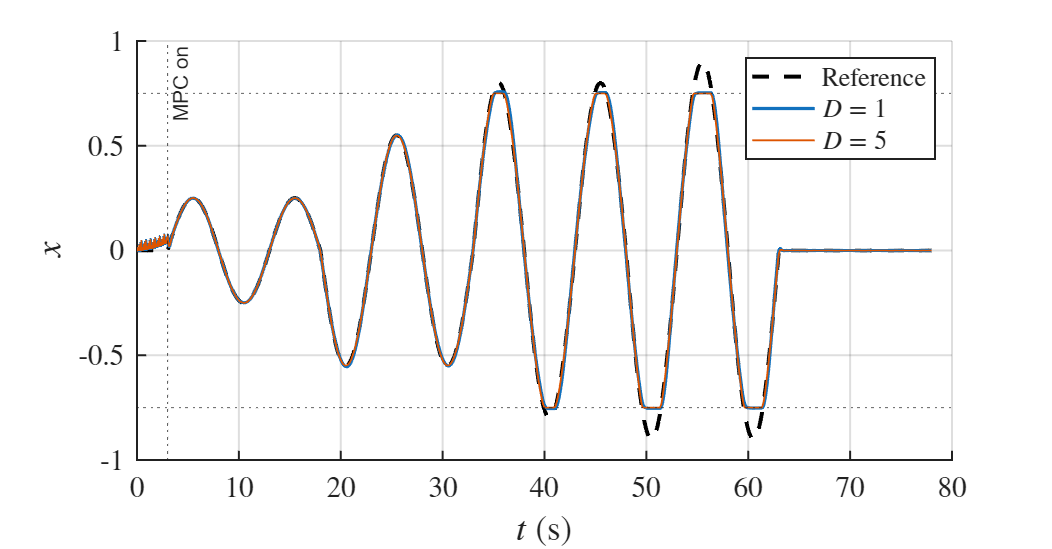}
    \caption{Closed-loop state tracking for the first- and fifth-order Taylor models. The horizontal dotted lines indicate the nominal bound $|x|\leq0.75$, which is softened with slack variables and may be violated.}
    \label{fig:tracking}
\end{figure}
\begin{figure}
    \centering
    \includegraphics[width=\linewidth, trim={0 0 0 0},clip]{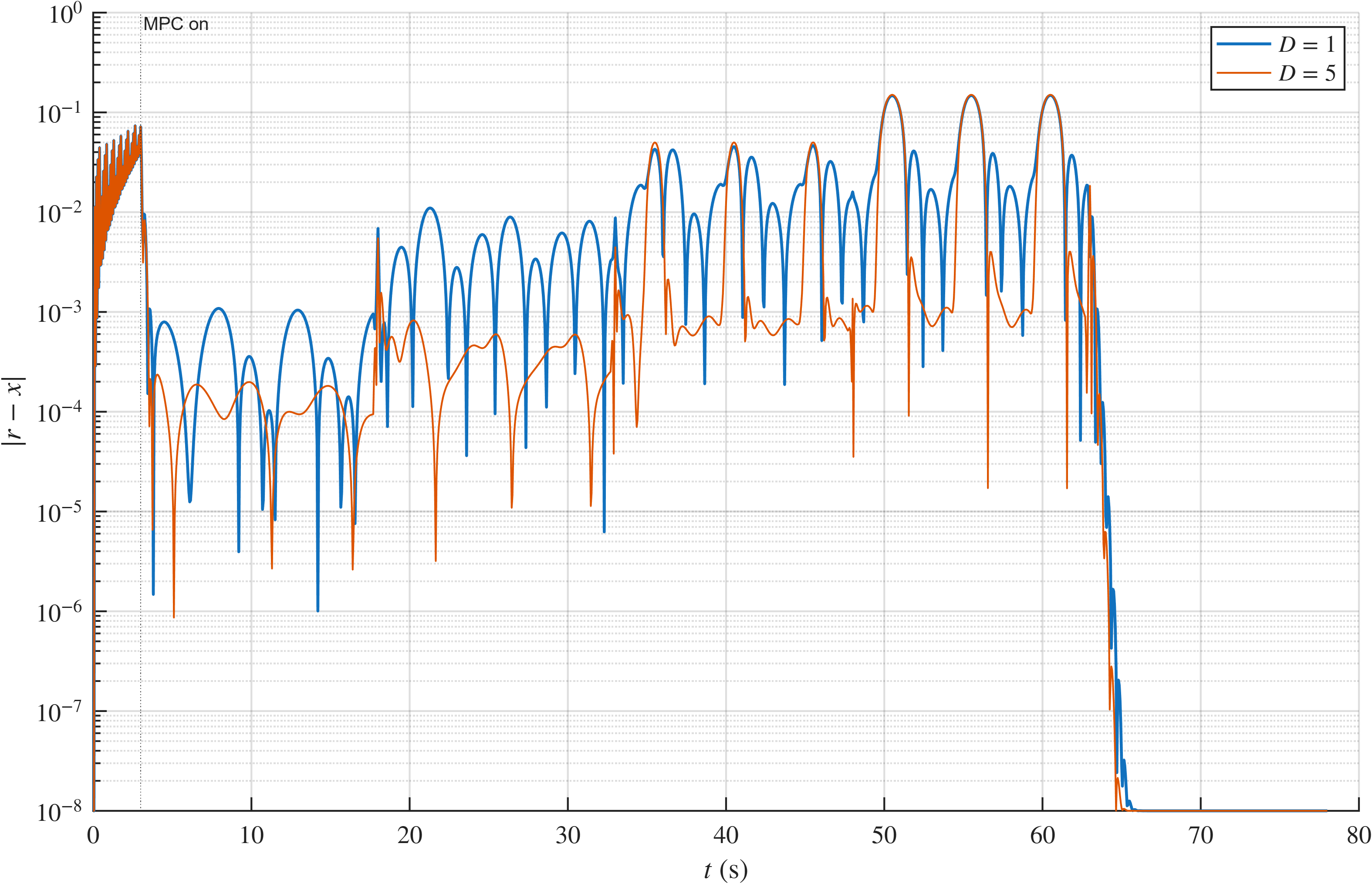}
    \caption{Absolute tracking error on a logarithmic vertical scale for the first- and fifth-order Taylor models. The plot includes the segment transients.}
    \label{fig:error_log}
\end{figure}
\begin{figure}
    \centering
    \includegraphics[width=\linewidth, trim={10 0 30 0},clip]{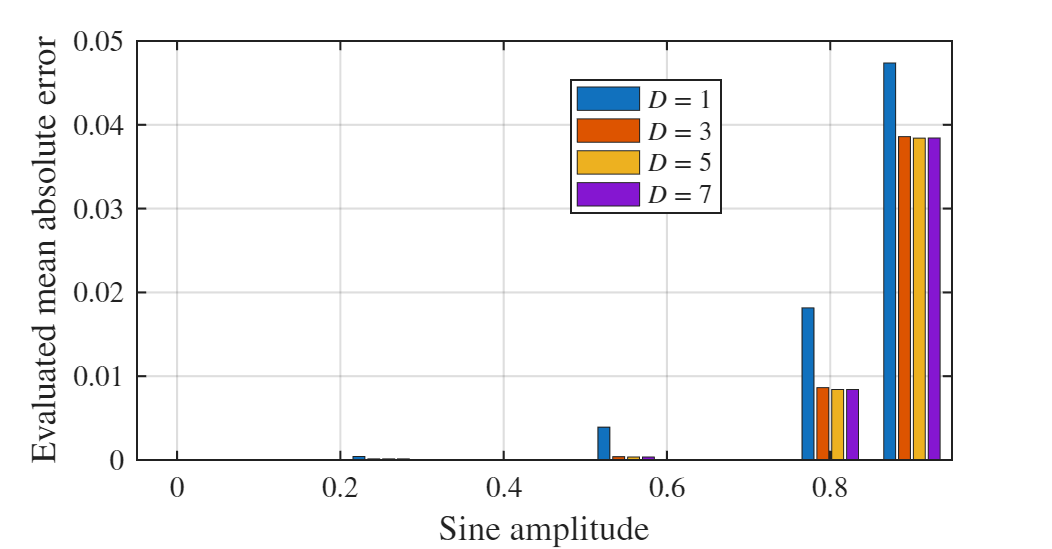}
    \caption{Evaluated post-transient mean absolute error (MAE) for each reference amplitude and Taylor degree. For every 300-sample amplitude segment, the first sinusoidal period (200 samples) is discarded and the MAE is evaluated over the remaining 100 samples.}
    \label{fig:amplitude_error}
\end{figure}
Figures~\ref{fig:tracking} and \ref{fig:error_log} show the state trajectories
and logarithmic tracking errors, respectively.

Each per-segment metric is
computed over its 300-sample half-open interval
$[\mathrm{segmentStart},\mathrm{segmentStop})$. For the evaluated
post-transient MAE in Figure~\ref{fig:amplitude_error}, the first 200 samples
are discarded and the remaining 100 are averaged. In contrast, the overall
RMSE and MAE in Table~\ref{tab:performance} use all stored samples
$k\geq N_{\mathrm{id}}$, including the final sample $k=K$.

RMSE and MAE measure overall and average tracking error, respectively. The
maximum state violation is the largest exceedance of the nominal soft state
bound; total input variation (I.V.) quantifies accumulated control variation.
\begin{table}[t]
\centering
\caption{Closed-loop overall performance for different Taylor degrees
$D$.}
\label{tab:performance}
\small
\setlength{\tabcolsep}{3pt}
\begin{tabular}{c|cccc}
\hline
$D$
& RMSE
& MAE
& \shortstack{Maximum state\\violation}
& \shortstack{Total input\\variation} \\
\hline
1 & 0.03199 & 0.01418 & $7.37\times10^{-3}$ & 12.41 \\
3 & 0.03092 & 0.00971 & $3.50\times10^{-4}$ & 11.75 \\
5 & 0.03090 & 0.00958 & $2.56\times10^{-4}$ & 11.86 \\
7 & 0.03089 & 0.00959 & $3.52\times10^{-4}$ & 11.81 \\
\hline
\end{tabular}
\end{table}

Figure~\ref{fig:tracking} compares the first- and fifth-order Taylor models.
Both responses track the reference, with visible clipping near the softened
state bound during the larger-amplitude segments. The nonzero violations in
Table~\ref{tab:performance} confirm that exact state-constraint satisfaction
is not claimed.

Figure~\ref{fig:error_log} presents the corresponding absolute tracking
errors, including the segment transients, on a logarithmic vertical scale.
This scale displays the small within-segment errors and the larger error
peaks in the same panel.

Figure~\ref{fig:amplitude_error} quantifies this trend through the evaluated
post-transient MAE for each sinusoidal amplitude after the 200-sample discard.
This metric is distinct from the overall MAE in
Table~\ref{tab:performance}. Differences among Taylor degrees are negligible
at small amplitudes. As the amplitude increases, the first-order model shows
consistently larger mean errors than the higher-order models, while the
third-, fifth-, and seventh-order models remain close to one another. This
behavior is consistent with higher-order nonlinear terms becoming more
influential as the system operates farther from the equilibrium about which
the Taylor expansion is constructed.

\section{Conclusion}

This paper presented a Taylor-informed indirect adaptive predictive control
framework for nonlinear systems based on Jacobian-frozen affine predictors.
In the numerical study, the continuous-time plant is simulated under
zero-order hold using RK4 with $T_s=0.05~{\rm s}$, and RLS identifies the
resulting one-step transitions. Online identification of the implemented
forward-Euler/Taylor-structure-informed reduced dictionary yields an adaptive
sampled nonlinear map. At each sampling instant, its Jacobian is evaluated at
the current operating point and frozen over the prediction horizon, producing
an affine predictor for finite-horizon model predictive control.

The results show lower overall MAE for $D\in\{3,5,7\}$ than for $D=1$,
with similar total input variation. Joint-odd symmetry exactly removes
even-total-degree monomials. Deliberate forward-Euler-informed pruning reduces
the model but does not imply that omitted monomials vanish in the RK4 sampled
flow. Maximum soft-state violations are reported.

Future work will investigate variable-rate forgetting for improved
adaptation under changing operating conditions, extend the framework to
systems with time delays \cite{Michiels2025}, and validate the approach on
real-world physical platforms.

\section*{Acknowledgment}

This research was partially supported by the MEXT SPReAD AI for Science program 2026 (Grant No. 26276994).

\bibliographystyle{IEEEtran}
\bibliography{bibliography.bib}

\end{document}